\documentclass[11pt,twoside]{article}
\usepackage{amsmath}
\usepackage{graphicx}
\usepackage{amsthm}
\usepackage{amssymb}

\newtheorem{thm}{Theorem}[section]

\newtheorem{lem}[thm]{Lemma}
\newtheorem{defn}[thm]{Definition}
\newtheorem{rem}[thm]{Remark}
\newtheorem{alg}[thm]{Algorithm}

\newtheorem{exa}[thm]{Example}

\title{\bf On the connection between cherry-tree copulas and truncated R-vine copulas}

\author{\\ Edith KOV\'{A}CS  
and Tam\'{a}s SZ\'{A}NTAI\\ Budapest University of Technology and Economics}

\date{}

\begin{document}

\thispagestyle{plain}

\maketitle

\begin{abstract}
Vine copulas are a flexible way for modeling dependences using only 
pair-copulas as building blocks. However if the number of variables grows the
problem gets fastly intractable. For dealing with this problem Brechmann at al.
proposed the truncated R-vine copulas. The truncated R-vine copula has the
very useful property that it can be constructed by using only pair-wise
copulas and a lower number of conditional pair-wise copulas. In our earlier papers we introduced
the concept of cherry-tree copulas. In this paper we characterize the relation
between the cherry-tree copulas and the truncated R-vine copulas. Both are
based on exploiting of some conditional independences between the variables.
We give a necessary and sufficient condition for a cherry-tree copula to be a
truncated R-vine copula.  
We introduce a new perspective for truncated R-vine modeling. The new idea is
finding first a good fitting cherry-tree copula of order $k$. Then, if this is
also a truncated R-vine copula we can apply the Backward Algorithm introduced in
this paper. This
way the construction of a sequence of trees which leads to it becomes possible. So
the cherry-tree copula can be expressed by pair-copulas and conditional
pair-copulas. In the case when the fitted $k$ order cherry-tree copula is not
a truncated R-vine copula we give an algorithm to transform it into  truncated
R-vine copula at level $k+1$. Therefore this cherry-tree copula can also be
expressed by pair-copulas.
\end{abstract}
\bigskip

\section{Introduction}
\label{sec:sec1}
\smallskip

Copulas in general are known to be useful tool for modeling
multivariate probability distributions since they make possible to model separately the dependence structure and the univariate marginals.
In this paper we show how conditional independences can be utilized in the expression of multivariate copulas.
Regarding to this we proved in Kov{\'a}cs and Sz{\'a}ntai (2013) a theorem which links to a junction tree probability distribution the so called junction tree copula.

The paper Aas et al. (2009) calls the attention on the fact that ''conditional independence may reduce the number of the
pair-copula decompositions and hence simplify the construction''. 
In this paper the importance of choosing a good
factorisation which takes advantage of the conditional independence
relations between the random variables is pointed out. 
In Kov{\'a}cs and Sz{\'a}ntai (2013) we introduced the concept of cherry-tree copulas which exploits the conditional
independences between the variables.

The importance of taking into account the conditional independences between the
variables encoded in a Bayesian Network (related to a directed acyclic graph) was
explored in the papers Kurowicka and Cooke (2002) and Hanea et al. (2006). Two aspects of this problem were
discussed. First, when the Bayesian Network (BN) is known, some of the conditional
independences taken from the BN are used to simplify a given expression of the D- or C-vine
copulas, which are very special types of vine copulas. Second, the problem of reconstruction of the BN from a sample data set was formulated 
under the assumption that the joint distribution is multivariate normal.
For discovering the independences and conditional independences between the
variables in Hanea et al. (2006) the correlations, the conditional
correlations and the determinant of the correlation matrix are used.

The paper Bauer et al. (2012) is dealing with more general  
pair-copula constructions related to the non-Gaussian BNs. In their paper the BN is
supposed to be known. The formula of probability distribution associated to
the given BN is expressed by pair-copulas.

The truncated Regular-vine (R-vine) copula is defined in Kurowicka (2011) and Brechmann et al. (2012). In Kurowicka (2011)  an algorithm is developed for searching the "best truncated vine", which is defined by Kurowicka (2011) as the one whose nodes of the top trees (trees with most conditioning) correspond to the smallest absolute values of correlations. This restricts the applicability of this method to Gaussian copulas. 

In this paper we recall the concept of the cherry-tree copulas. An alternative definition of R-vines using a special hypergraph structure is given in Kov{\'a}cs and Sz{\'a}ntai (2013). There we proved that truncated vine copula is a special case of the cherry-tree copula.

In the preliminary section we recall all concepts that we need in the paper.
First we will remind some graph theoretical concepts since the conditional independences can be represented on graphs.
Then the concepts of copulas and R-vine copulas will be recalled. Finally the multivariate junction tree copula associated to a junction tree probability distribution and an equivalent definition of the R-vine copulas based on the cherry-tree graph structures will also be presented. 
In the third section we give a necessary and sufficient condition for a cherry-tree copula to be a truncated R-vine copula and an algorithm for obtaining the truncated R-vine structure from a given cherry-tree copula. 
In the fourth section we give a theorem for transforming a general cherry-tree copula into a truncated R-vine copula.
We finish the paper with conclusions and with highlighting the new perspectives given by the paper.

\section{Preliminaries}
\label{sec:sec2}
\smallskip

The reader who is familiar with the basic concepts presented in this preliminary section 
may skip some parts of it.

\subsection{Acyclic hypergraph, junction tree, junction tree probability distribution}
\label{subsec:subsec2.1}

In this subsection we introduce some concepts used in graph theory and
probability theory which are needed throughout the paper and present how these
can be linked to each other. For a good overview see Lauritzen and Spiegelhalter (1988) and Lauritzen (1996).
We first present the acyclic hypergraphs and junction trees. 
Then we introduce the cherry-trees as a special type of junction trees.
We finish this subsection with the multivariate joint
probability distribution associated to junction trees.

Let $V=\left\{ 1,\ldots ,d\right\} $ be a set of vertices and $\Gamma $ a
set of subsets of $V$ called {\it set of hyperedges}. A \textit{hypergraph}
consists of a set $V$ of vertices and a set $\Gamma $ of hyperedges. We
denote a hyperedge by $K_i$, where $K_i$ is a subset of $V$. 

The \textit{acyclic}{\normalsize \ }\textit{hypergraph} is a special type of
hypergraph which fulfills the following requirements:

\begin{itemize}
\item  Neither of the hyperedges of $\Gamma $ is a subset of another hyperedge.

\item  There exists a numbering of edges for which the \textit{running
intersection property} is fulfilled: $\forall j\geq 2\quad \ \exists \ i<j:\
K_i\supset K_j\cap \left( K_1\cup \ldots \cup K_{j-1}\right) $. (Other
formulation is that for all hyperedges $K_i$ and $K_j$ with $i<j-1$,
$K_i \cap K_j \subset K_s \ \mbox{for all} \ s, i<s<j$.)
\end{itemize}

Let $S_j=K_j\cap \left( K_1\cup \ldots \cup K_{j-1}\right) $, for $j>1$ and $%
S_1=\phi $. Let $R_j=K_j\backslash S_j$. We say that $S_j$\textit{separates} 
$R_j$ from $\left( K_1\cup \ldots \cup K_{j-1}\right) \backslash S_j$, and
call $S_j$ separator set or shortly separator.

Now we link these concepts to the terminology of junction trees.

The {\it junction tree} is a special tree stucture which is equivalent to the
connected acyclic hypergraphs (Lauritzen and Spiegelhalter (1988)) and Lauritzen (1996). The nodes of the tree correspond
to the hyperedges of the connected acyclic hypergraph and are called clusters, the "edges" of the junction tree
correspond to the separator sets and called separators. The set of all
clusters is denoted by $\Gamma$, the set of all separators is denoted by %
$\mathcal{S}$. A junction tree $(V,\Gamma, S)$ is defined by the set of vertices $V$, the set of nodes $\Gamma$ called also set of clusters, and the set of separators $S$. The junction tree with the largest cluster containing $k$
variables is called \textit{k-width junction tree}. 

An important relation between graphs and hypergraphs is given in Lauritzen and Spiegelhalter (1988): A
hypergraph is acyclic if and only if it can be considered to be the set of maximal
cliques of a triangulated ({\it chordal}) graph (a graph is triangulated if every cycle of
length greater than 4 has a chord). This means that the vertices in a cluster are all connected with each other.

In the Figure \ref{fig:fig1} one can see a) a triangulated graph, b) the
corresponding acyclic hypergraph and c) the corresponding junction tree.

\begin{figure}
  \includegraphics[bb=80 590 420 770,width=10cm]{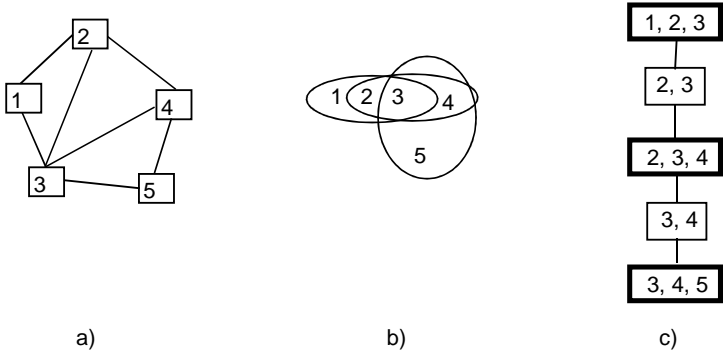}
\caption{a) Triangulated graph, b) The corresponding acyclic hypergraph, c) The corresponding junction tree which is a cherry-tree}
\label{fig:fig1}       
\end{figure}

We consider the random vector $\mathbf{X}=\left( X_1,\ldots ,X_d\right) ^T$, with the
set of indices $V=\left\{ 1,\ldots ,d\right\} $. 

Major advances in probabilistic inference methods based on graphical representations have been realized by Lauritzen and Spiegelhalter (1998) and Lauritzen (1996).
However probabilistic inference has the inherent disadvantage of being NP-hard. By exploiting the conditional independence relations among the discrete random variables of a probabilistic network the underlying joint probability space maybe decomposed into smaller subspaces corresponding to cliques in a triangulated graph (Lauritzen and Spiegelhalter (1998) and Lauritzen (1996)).

We say that the probability distribution associated to a graph has the Global Markov (GM) property if in the graph $\forall A,B,C\subset V$ and C separates A and B in terms of graph then $\mathbf{X}_A$ and $\mathbf{X}_B$ are conditionally independent given $\mathbf{X}_C$, which means in terms of probabilities that 

\[P\left( \mathbf{X}_{A\cup B\cup C}\right) =\dfrac{P\left( \mathbf{X}_{A\cup
C}\right) P\left( \mathbf{X}_{B\cup C}\right) }{P\left( \mathbf{X}_C\right) }.
\]

The concept of \textit{junction tree probability distribution} is related to
the junction tree graph and to the Global Markov property. A
junction tree probability distribution is defined as a fraction of some products of marginal probability distributions as follows:

\begin{equation}\label{eq:eq0}
P\left( \mathbf{X}\right) =\dfrac{\prod\limits_{C\in \Gamma}P\left(
\mathbf{X}_C\right) }{\prod\limits_{S\in \mathcal{S}}\left[ P\left( \mathbf{X}_S\right)
\right] ^{\nu _S-1}},
\end{equation}
where $\Gamma$ is the set of clusters of the
junction tree, $\mathcal S$ is the set of separators, $\nu_S$ is the number of those clusters
which are linked by the separator $S$. We note here that any equality of type (\ref{eq:eq0}) means that it is fulfilled 
for any possible realization of $\mathbf{X}$.

\begin{exa}
\label{ex:ex2.1}
The probability distribution corresponding to Figure \ref{fig:fig1} is:
\[
\begin{array}{rcl}
P\left( \mathbf{X}\right) & = & \dfrac{P\left( \mathbf{X}_{\{1,2,3\}}\right) P\left(
\mathbf{X}_{\{2,3,4\}}\right) P\left(\mathbf{X}_{\{3,4,5\}}\right) }{P\left(\mathbf{X}_{\{2,3\}}\right)
P\left( \mathbf{X}_{\{3,4\}}\right) }\vspace{3mm} \\
 & = & \dfrac{P\left( X_1,X_2,X_3\right) P\left(
X_2,X_3,X_4\right) P\left( X_3,X_4,X_5\right) }{P\left( X_2,X_3\right)
P\left( X_3,X_4\right) }.
\end{array}
\]
\end{exa}

In Buksz\'{a}r and Pr\'{e}kopa (2001) and Buksz\'{a}r and Sz\'{a}ntai (2002) there were used and named the so called $t$-cherry-tree graph structures. Since these can be regarded as a special type of junction tree we can give now the following definition. In this paper we will call this structure simply cherry-tree as this does not cause any confusion.

\begin{defn}
\label{def:def2.2} 
We call $k$ order cherry-tree the junction tree with all clusters of size $k$ and all separators of size $k-1$.
\end{defn}

Denoting by $\mathcal{C}_{\mbox{ch}}$ and $\mathcal{S}_{\mbox{ch}}$ the
set of clusters and the set of separators of the cherry junction tree, respectively we gave the following definition. 

\begin{defn}
\label{def:def2.3}
(Kov\'{a}cs, Sz\'{a}ntai (2012a))
In the discrete case the probability distribution given by (\ref{eq:eq1}) is called {\it cherry-tree probability distribution}

\begin{equation}\label{eq:eq1}
P_{\mbox{t-ch}}(\mathbf{X})=\dfrac{\prod\limits_{K\in \mathcal{C}_{\mbox{ch}}}P\left( \mathbf{X}_K\right) }{%
\prod\limits_{S\in \mathcal{S}_{\mbox{ch}}}\left( P\left( \mathbf{X}_S\right) \right) ^{\nu _s-1}}
\end{equation}
and in the continuous case the probability distribution given by (\ref{eq:eq2}) is called {\it cherry-tree probability density function}
\begin{equation}\label{eq:eq2}
f_{t-ch}\left( \mathbf{x}\right) =\dfrac{\prod\limits_{K\in C_{ch}}f_K\left( 
\mathbf{x}_k\right) }{\prod\limits_{S\in S_{Ch}}\left( f_S\left( \mathbf{x}%
_k\right) \right) ^{\nu _S-1}},
\end{equation}
where in both cases $\nu_S$ denotes the number of clusters which are linked by the separator $S$.
\end{defn}

The marginal probability distributions and the marginal density functions involved in the above
formulae are marginal probability distributions of $P\left( \mathbf{X}\right) $, respectively marginal density functions of $f\left( \mathbf{x}\right) $.

\subsection{Copula, R-vine copula}
\label{subsec:subsec2.2}

In this subsection we recall some definitions according to copulas and R-vine copulas.

\begin{defn}\label{def:def2.7}
A function $C:\left[0;1\right]^{d}\rightarrow\left[0;1\right]$
is called a $d$-dimensional copula if it satisfies the following conditions:

\begin{enumerate}

\item $C\left(u_{1},\ldots,u_{d}\right)$ is increasing in each component
$u_{i}$,

\item $C\left(u_{1},\ldots,u_{i-1},0,u_{i+1},\ldots,u_{d}\right)=0$ for all
$u_{k}\in\left[0;1\right]$,\ $k\neq i,\ i=1,\ldots,n$,

\item $C\left(  1,\ldots,1,u_{i},1,\ldots,1\right)  =u_{i}$ for all $u_{i}%
\in\left[  0;1\right]  ,\ i=1,\ldots,d$,

\item C is $d$-increasing, i.e for all $\left(  u_{1,1},\ldots,u_{1,d}\right)  $
and $\left(  u_{2,1},\ldots,u_{2,d}\right)  $ in $\left[  0;1\right]  ^{d}$
with $u_{1,i}<u_{2,i}$ for all i, we have

\[
\sum\limits_{i_{1}=1}^{2}\cdots
\sum\limits_{i_{d}=1}^{2}\left(  -1\right)  ^{\sum\limits_{j=1}^{d}i_{j}%
}C\left(  u_{i_{1},1},\ldots,u_{i_{d},d}\right)  \geq0.
\]

\end{enumerate}
\end{defn}

Due to Sklar's theorem if $X_{1},\ldots,X_{d}$ are continuous random variables
defined on a common probability space, with the univariate marginal cdf's
$F_{X_{i}}\left(  x_{i}\right)  $ and the joint cdf $F_{X_{1},\ldots,X_{d}%
}\left(  x_{1},\ldots,x_{d}\right)  $ then there exists a unique copula
function $C_{X_{1},\ldots,X_{d}}\left(  u_{1},\ldots,\right.$ 
$\left. u_{d}\right)  :\left[
0;1\right]  ^{d}\rightarrow\left[  0;1\right]  $ such that by the substitution
$u_{i}=F_{i}\left(  x_{i}\right), \ i=1,\ldots,d$ we get

\[
F_{X_{1},\ldots,X_{d}}\left(x_{1},\ldots,x_{d}\right)
=C_{X_{1},\ldots,X_{d}}\left(F_{1}\left(  x_{1}\right)  ,\ldots,F_{d}\left(  x_{d}\right)  \right)
\]
for all $\left(x_{1},\ldots,x_{d}\right)^{T}\in R^{d}.$

In the following we will use the vectorial notation
$F_{\mathbf{X}_{D}}\left(\mathbf{x}_{D}\right)=C_{\mathbf{X}_{D}}\left(
\mathbf{u}_{D}\right)$, where $\mathbf{u}_{D}%
=\left(  F_{X_{i_{1}}}\left(  x_{i_{1}}\right)  ,\ldots,F_{X_{i_{m}}}\left(
x_{i_{m}}\right)  \right)  ^{T}$ and $\left\{ i_1,\ldots, i_m \right\}=D\subseteq V$.

It is well known that

\[
f_{X_{i_{1}},\ldots X_{i_{m}}}\left(  x_{i_{1}},\ldots,x_{i_{m}}\right)
=c_{X_{i_{1}},\ldots X_{i_{m}}} \left( F_{X_{i_{1}}}\left(x_{i_{1}}\right),\ldots,F_{X_{i_{m}}}\left(  x_{i_{m}}\right)  \right)
\cdot \prod\limits_{k=1}^{m} f_{X_{i_{k}}} \left(x_{i_{k}}\right)
\]

In vectorial terms this can be written as
\begin{equation}\label{eq:eq3}
f_{\mathbf{X}_{D}}\left(  \mathbf{x}_{D}\right)  =c_{\mathbf{X}_{D}}\left(
\mathbf{u}_{D}\right)  \cdot\prod\limits_{i_{k}\in D}^{{}%
}f_{X_{i_{k}}}\left(  x_{i_{k}}\right),
\end{equation}
and after dividing by the product term
\[
c_{\mathbf{X}_{D}}\left(  \mathbf{u}_{D}%
\right)  =\dfrac{f_{\mathbf{X}_{D}}\left(  \mathbf{x}_{D}\right)  }%
{\prod\limits_{i_{k}\in D}^{{}}f_{X_{i_{k}}}\left(  x_{i_{k}}\right)  }.
\]

Pair-copula construction introduced by Joe (1997) is able to encode more
types of dependences in the same time since they can be expressed as a product of different types of bivariate copulas.
The R-vine structures were introduced by Bedford and Cooke (2001, 2002) and described in more details by Kurowicka and Cooke (2006). 

If it does not cause confusion, instead of $f_{\mathbf{X}_D}$ and $c_{\mathbf{X}_D}$  
we will write $f_D$ and $c_D$. We also introduce the following notations:

\begin{tabular}{lcl}
$F_{i,j|D}$ & -- & the conditional probability distribution function of $X_i$ and $X_j$ \\
            &    & given $\mathbf{X}_D$;\\
$f_{i,j|D}$ & -- & the conditional probability density function of $X_i$ and $X_j$ \\
            &    & given $\mathbf{X}_D$, \\
$c_{i,j|D}$ & -- & the conditional copula density function corresponding to  $f_{i,j|D}$,\\
\end{tabular}

where $D\subset V;i,j\in V\backslash D$.

According to Kurowicka and Cooke (2006), the definition of the R-vine graph structure is given as:
\begin{defn}\label{def:def2.8}
A \textit{Regular-vine (R-vine) on d variables} consists of a
sequence of trees $T_1,T_2,\ldots ,T_{d-1}$ with nodes $N_i$ and edges $E_i$
for $i=1,\ldots ,d-1$, which satisfy the following conditions:
\end{defn}

\begin{itemize}
\item  $T_1$ has nodes $N_1=\left\{ 1,\ldots ,d\right\} $ and edges $E_1$.

\item  For $i=2,\ldots ,d-1$ the tree $T_i$ has nodes $N_i=E_{i-1}$.

\item  Two edges $a$ and $b$ in tree $T_i$ are joined in tree $T_{i+1}$ if $\{a,b\} \in E_{i}$ and $|a \Delta b| = 2$, where $\Delta$ stands for the symmetric difference operator and $| \cdot |$ stands for the cardinality of the set.
\end{itemize}

We mention here that $a$ and $b$ are subsets of $V$ and $|a|=|b|=i$ in the tree $T_i$.

The last condition usually is referred to as \textit{proximity condition}.

It is shown in Bedford and Cooke (2001) and Kurowicka and Cooke (2006) that the edges in an R-vine
tree can be uniquely identified by two nodes (the conditioned nodes), and a
set of conditioning nodes, i.e., edges are denoted by $e=j\left( e\right)
,k\left( e\right) |D\left( E\right) $ where $D\left( E\right) $ is the
conditioning set and $j\left( e\right),k\left( e\right)\notin D(E)$. For a good overview see Czado (2010). 

The next theorem which can be regarded as a central theorem of R-vines see Bedford and Cooke (2001) links the
probability density function to the copulas assigned to the R-vine
structure. In Bedford and Cooke (2001) it is shown that there exists a unique probability density assigned to a given R-vine structure. In Bedford and Cooke (2002) it is shown that this probability distribution can be expressed as (\ref{eq:eq4}).

\begin{thm}\label{theo:theo2.9}
The joint density of $\mathbf{X}=\left( X_1,\ldots ,X_d\right) $ is
uniquely determined and given by:
\begin{equation}\label{eq:eq4}
\begin{array}{l}
f\left( x_1,\ldots ,x_d\right)=\left[ \prod\limits_{k=1}^df_k\left(x_k\right) \right]\\
\cdot\prod\limits_{i=2}^{d-1}\prod\limits_{e\in E_i}c_{j\left(
e\right) ,k\left( e\right) |D\left( e\right) }\left( F_{j\left( e\right)
|D\left( e\right) }\left( x_{j\left( e\right) }|\mathbf{x}_{D\left( e\right)
}\right) ,F_{k\left( e\right) |D\left( e\right) }\left( x_{k\left( e\right)
}|\mathbf{x}_{D\left( e\right) }\right)\left | \mathbf{x}_{D\left( e\right) } \right .\right),
\end{array}
\end{equation}
where $F_{j\left( e\right) |D(e)}\left( x_{j\left( e\right) }|\mathbf{x}%
_{D\left( e\right) }\right) $ can be calculated as follows:
\[
F_{j\left( e\right) |D(e)}\left( x_{j\left( e\right) }|\mathbf{x}_{D\left(
e\right) }\right) =\frac{\partial C_{i,j(e)|D\left( e\right) \backslash
\left\{ i\right\} }\left( u_i,u_j\right) }{\partial u_i}\left |_
{\hspace{-7mm}u_i=F_{i|D\left( e\right) \backslash \left\{ i\right\} }\left( x_i|\mathbf{x}%
_{D\left( e\right) \backslash \left\{ i\right\} }\right)  \atop %
u_j=F_{j(e)|D\left( e\right) \backslash \left\{ i\right\} }\left( x_{j\left(
e\right) }|\mathbf{x}_{D\left( e\right) \backslash \left\{ i\right\}
}\right)}\right .
\]
for $i\in D\left( e\right)$.
\end{thm}


Thus one can express $F_{j\left( e\right) |D(e)}\left( x_{j\left( e\right) }|%
\mathbf{x}_{D\left( e\right) }\right) $as a function of two conditional
distributions $F_{i|D\left( e\right) \backslash \left\{ i\right\} }\left(
x_i|\mathbf{x}_{D\left( e\right) \backslash \left\{ i\right\} }\right) $and $%
F_{j(e)|D\left( e\right) \backslash \left\{ i\right\} }\left( x_{j\left(
e\right) }|\mathbf{x}_{D\left( e\right) \backslash \left\{ i\right\}
}\right) $, with one less conditioning variable. This formula was given by Joe (1997). Hence all conditional
distribution functions in (\ref{eq:eq4}) are nested functions of the
univariate marginal distribution functions.
In (\ref{eq:eq4}) only pair-copulas are involved, therefore these constructions are also called pair-copula constructions.

We emphasize here that in general the parameter of the pair-copulas 
$c_{j\left( e \right) ,k\left( e\right) |D\left( e\right) }$ depends on the conditioning set $D(e)$.

In Subsection \ref{subsec:subsec2.4} we will give another definition for the R-vine which is related to a sequence of $k$ order cherry-trees.

In Aas et al. (2009) the inference of pair-copula decomposition is depicted in three parts: 

\begin{itemize}
\item  The selection of a specific factorization (structure);
\item  The choice of pair-copula types;
\item  The estimation of parameters of the chosen pair-copulas.
\end{itemize}

Our approach deals with finding a good factorization which exploits some of the conditional independences between the random variables.

There are many papers dealing with selecting specific
R-vines as C-vine or D-vine see for example Aas et al. (2009).

\subsection{The multivariate copula associated to a junction tree probability distribution. The cherry-tree copulas.}
\label{subsec:subsec2.3}
\smallskip

In this subsection we recall results published in Kov{\'a}cs, Sz{\'a}ntai (2012a) and (2012b). In Kov{\'a}cs, Sz{\'a}ntai (2012a) a theorem assuring the existence of a special type of copula density was proved. This copula density was assigned to a junction tree graph structure.
Let us consider a random vector $\mathbf{X}=(X_1, X_2, \ldots, X_d)^{T}$, with the set of indices $V=\{1, 2, \ldots, d\}$.  Let $(V,\Gamma,S)$  be a junction tree defined on the vertex set $V$, by the cluster set $\Gamma$, and the separator set $S$.

\begin{thm}\label{theo:theo3.1} (Kov{\'a}cs, Sz{\'a}ntai (2012a))
The copula density function associated to a junction tree
probability distribution
\[
f_{\mathbf{X}}\left(  \mathbf{x}\right)
=\dfrac{\prod\limits_{K\in\Gamma}f_{\mathbf{X}_{K}}\left(  \mathbf{x}%
_{K}\right)  }{\prod\limits_{S\in\mathcal{S}}\left[  f_{\mathbf{X}_{S}}\left(
\mathbf{x}_{S}\right)  \right]  ^{\nu_{S}-1}},
\]
is given by

\begin{equation}\label{eq:eq5}
c_{\mathbf{X}}\left(  \mathbf{u}_{V}\right)  = \dfrac{\prod\limits_{K\in\Gamma
}c_{\mathbf{X}_{K}}\left(  \mathbf{u}_{K}\right)  }{\prod\limits_{S\in
\mathcal{S}}\left[  c_{\mathbf{X}_{S}}\left(  \mathbf{u}_{S}\right)  \right]
^{\nu_{S}-1}},
\end{equation}
where $\nu_S$ is the number of clusters linked by $S$.
\end{thm}

\begin{defn}\label{def:def3.2} (Kov{\'a}cs, Sz{\'a}ntai (2012a))
The copula defined by (\ref{eq:eq5}) is called junction tree copula.
\end{defn}

We saw that if the conditional independence structure between the random
variables makes possible the construction of a junction tree, then the
multivariate copula density associated to the joint probability distribution
can be expressed as a fraction of some products of lower dimensional copula densities.

\begin{defn}\label{def:def3.3} (Kov{\'a}cs, Sz{\'a}ntai (2012a))
The copula density function associated to a cherry-tree probability distribution is called cherry-tree copula and its expression is: 
\begin{equation}\label{eq:eq8}
c_{\mathbf{X}}\left(  \mathbf{u}_{V}\right)  = \dfrac{\prod\limits_{K \in \mathcal{C}_{ch}
}c_{\mathbf{X}_{K}}\left(  \mathbf{u}_{K}\right)  }{\prod\limits_{S\in
\mathcal{S}_{ch}}\left[  c_{\mathbf{X}_{S}}\left(  \mathbf{u}_{S}\right)  \right]
^{\nu_{S}-1}},
\end{equation}
where $\nu_S$ denotes the number of clusters which contain the separator $S$.
\end{defn}

\subsection{R-vine structure given by a sequence of cherry-trees. Cherry-vine copula.}
\label{subsec:subsec2.4}

In Kov{\'a}cs, Sz{\'a}ntai (2012a) and (2012b) we gave an alternative definition for R-vines by using the concept of cherry-tree.

\begin{defn}\label{def:def4.1}
The cherry-vine graph structure is defined by a sequence of 
cherry junction trees $T_1,T_2,\ldots ,T_{d-1}$ as follows

\begin{itemize}
\item  $T_1$is a regular tree on $V=\left\{ 1,\ldots ,d\right\} $, the set
of edges is $E_1=\left\{ e_i^1=\left( l_i,m_i\right)\right. ,$ $\left. i=1,\ldots ,d-1,\
l_i,m_i\in V\right\}$

\item  $T_2$ is the second order cherry junction tree on $%
V=\left\{ 1,\ldots ,d\right\} $, with the set of clusters $E_2=\left\{
e_i^2,i=1,\ldots ,d-1|e_i^2=e_i^1\right\} $ , $\left| e_i^1\right| =2$

\item  $T_k$ is one of the possible $k$ order cherry junction tree on $%
V=\left\{ 1,\ldots ,d\right\}$, with the set of clusters $E_k=\left\{
e_i^k,i=1,\ldots ,d-k+1\right\}$ , where each $e_i^k,\left| e_i^k\right| =k$
is obtained from the union of two linked clusters in the $\left( k-1\right) $
order cherry junction tree $T_{k-1}$.
\end{itemize}
\end{defn}

It is straightforward to see that Definition \ref{def:def4.1} is equivalent with Definition \ref{def:def2.8}. Next we define the pair-copulas assigned to the cherry-vine structure given in Definition \ref{def:def4.1}

The copula densities $c_{l_i,m_i}\left( F_{l_i}\left(
x_{_{li}}\right) ,F_{m_i}\left( x_{_{m_i}}\right) \right)$ are assigned to
the edges of the tree $T_{1}$.


The copula densities
\[
c_{a_{ij}^l,b_{ij}^l|S_{ij}^l}\left( F_{a_{ij}^l|S_{ij}^l}\left(
x_{a_{ij}^l}|\mathbf{x}_{S_{ij}^l}\right) ,F_{b_{ij}^l|S_{ij}^l}\left(
x_{b_{ij}^l}|\mathbf{x}_{S_{ij}^l}\right) \left | \;\mathbf{x}_{S_{ij}^l} \right. \right) 
\]
are assigned to each pair of clusters $e_i^l$ and $e_j^l$, which are linked
in the junction tree $T_l$, where:
\begin{equation}\label{eq:eq7a}
\begin{array}{rcl}
S^l &=&e_i^l\cap e_j^l,\vspace{2mm}\\
a_{ij}^l &=&e_i^l-S_{ij}^l \vspace{2mm}\\
b_{ij}^l &=&e_i^l-S_{ij}^l,
\end{array}
\end{equation}
for $l=2,\ldots , d-1$.
It is easy to see that $a_{ij}^l$ and $b_{ij}^l, \; l=2,\ldots , d-1$ contain a single element only.

The following theorem is a consequence of Theorem \ref{theo:theo2.9}.

\begin{thm}\label{theo:theo4.2}
The probability distribution associated to the cherry-vine structure given in Definition \ref{def:def4.1} can be expressed as:
\begin{equation}\label{eq:eq9}
\begin{array}{l}
f\left(  x_{1},\ldots,x_{d}\right)  =\left[  \prod\limits_{i=1}^{d}%
f_{i}\left(  x_{i}\right)  \right]  \left[  \prod\limits_{\left(  l_{i}%
m_{i}\right)  \in E_{1}}^{{}}c_{l_{i}m_{i}}\left(  F_{l_{i}}\left(  x_{l_{i}%
}\right)  ,F_{m_{i}}\left(  x_{m_{i}}\right)  \right)  \right]  \cdot\\
\prod\limits_{l=2}^{d-1}\prod\limits_{e_{i}^{l},e_{j}^{l}\in N\left(
T_{l}\right)  }^{{}}c_{a_{i,j}^{l},b_{i,j}^{l}|S_{ij}^{l}}
\left(F_{a_{i,j}^{l}|S_{ij}^{l}}\left(  x_{a_{i,j}^{l}}|\mathbf{x}_{S_{ij}^{l}%
}\right)  ,F_{b_{i,j}^{l}|S_{ij}^{l}}\left(  x_{b_{i,j}^{l}}|\mathbf{x}%
_{S_{ij}^{l}}\right)  |\mathbf{x}_{S_{ij}^{l}}\right)\end{array}
\end{equation}
where $e_{i}^{l},e_{j}^{l}\in N\left(  T_{l}\right)  $ denotes that $e_{i}%
^{l},e_{j}^{l}$ are linked in the cherry tree $T_{l}$, and $S_{ij}^{l}%
,a_{i,j}^{l},b_{i,j}^{l},$ are defined by (\ref{eq:eq7a}) and $F_{a_{i,j}^{l}|S_{ij}^{l}}%
$is defined in similar way as in Theorem \ref{theo:theo2.9}.
\end{thm}



We illustrate these concepts on the following example.

\begin{exa}
\label{ex:ex2.2}
The edge set of the first tree and the sequence of the cherry-trees (in Figure \ref{fig:fig2}) together with the copula densities determined by Definition \ref{def:def4.1} are the following:

\[
\begin{array}{rl}
T_1: & E_1=\left\{ \left( 1,2\right), \left( 2,3\right), \left(2,6\right), \left( 3,4\right), \left( 4,5\right) \right\}, \\
& c_{1,2},c_{2,3},c_{2,6},c_{3,4},c_{4,5};\\
T_2: & E_2=\left\{e_1^2=\left(1,2\right), e_2^2=\left(2,3\right),e_3^2=\left(2,6\right), e_4^2=\left(3,4\right), 
                  e_5^2=\left(4,5\right)\right\}\\
& S_{1,2}^2=e_1^2\cap e_2^2=\left\{ 2\right\},\\
& \hspace{10mm}a_{1,2}^2=e_1^2-S_{1,2}^2=\left\{ 1\right\},
  b_{1,2}^2=e_2^2-S_{1,2}^2=\left\{ 3\right\},
  c_{a_{1,2}^2,b_{1,2}^2|S_{1,2}^2}=c_{1,3|2}\\
& S_{2,3}^2=e_2^2\cap e_3^2=\left\{ 2\right\},\\
& \hspace{10mm}a_{2,3}^2=e_2^2-S_{2,3}^2=\left\{ 3\right\},
  b_{2,3}^2=e_2^2-S_{2,3}^2=\left\{ 6\right\},
  c_{a_{2,3}^2,b_{2,3}^2|S_{2,3}^2}=c_{3,6|2}\\
& S_{2,4}^2=e_2^2\cap e_4^2=\left\{ 3\right\},\\ 
& \hspace{10mm}a_{2,4}^2=e_2^2-S_{2,4}^2=\left\{ 2\right\},
  b_{2,4}^2=e_4^2-S_{2,4}^2=\left\{ 4\right\},
  c_{a_{2,4}^2,b_{2,4}^2|S_{2,4}^2}=c_{2,4|3}\\
& S_{4,5}^2=e_4^2\cap e_5^2=\left\{ 4\right\},\\
& \hspace{10mm}a_{4,5}^2=e_4^2-S_{4,5}^2=\left\{ 3\right\},
  b_{4,5}^2=e_5^2-S_{4,5}^2=\left\{ 5\right\},
  c_{a_{4,5}^2,b_{4,5}^2|S_{4,5}^2}=c_{3,5|4};\\
T_3: & E_3=\left\{ e_1^3=\left( 1,2,3\right) ,e_2^3=\left( 2,3,4\right)
                  ,e_3^3=\left( 2,3,6\right) ,e_4^3=\left( 3,4,5\right) \right\}\\
& S_{1,2}^3=e_1^3\cap e_2^3=\left\{ 2,3\right\},\\
& \hspace{10mm}a_{1,2}^3=e_1^3-S_{1,2}^3=\left\{ 1\right\},
  b_{1,2}^3=e_2^3-S_{1,2}^3=\left\{ 4\right\},
  c_{a_{1,2}^3,b_{1,2}^3|S_{1,2}^3}=c_{1,4|2,3}\\ 
& S_{2,3}^3=e_2^3\cap e_3^3=\left\{ 2,3\right\},\\
& \hspace{10mm}a_{2,3}^3=e_2^3-S_{2,3}^3=\left\{ 4\right\},
  b_{2,3}^3=e_3^3-S_{2,3}^3=\left\{ 6\right\},
  c_{a_{2,3}^3,b_{2,3}^3|S_{2,3}^3}=c_{4,6|2,3}\\  
& S_{2,4}^3=e_2^3\cap e_4^3=\left\{ 3,4\right\},\\
& \hspace{10mm}a_{2,4}^3=e_2^3-S_{2,4}^3=\left\{ 2\right\},
  b_{2,4}^3=e_4^3-S_{2,4}^3=\left\{ 5\right\},
  c_{a_{2,4}^3,b_{2,4}^3|S_{2,4}^3}=c_{2,5|3,4};\\  
T_4: & E_4=\left\{ e_1^4=\left( 1,2,3,4\right) ,e_2^4=\left(2,3,4,5\right) ,e_3^4=\left( 2,3,4,6\right) \right\}\\
& S_{1,2}^4=e_1^4\cap e_2^4=\left\{ 2,3,4\right\},\\
& \hspace{10mm}a_{1,2}^4=e_1^4-S_{1,2}^4=\left\{ 1\right\},
  b_{1,2}^4=e_2^4-S_{1,2}^4=\left\{ 5\right\},
  c_{a_{1,2}^4,b_{1,2}^4|S_{1,2}^4}=c_{1,5|2,3,4}\\
& S_{2,3}^3=e_2^4\cap e_3^4=\left\{ 2,3,4\right\},\\
& \hspace{10mm}a_{2,3}^4=e_2^4-S_{2,3}^4=\left\{ 5\right\},
  b_{2,3}^4=e_3^4-S_{2,3}^4=\left\{ 6\right\},
  c_{a_{2,3}^4,b_{2,3}^4|S_{2,3}^4}=c_{5,6|2,3,4}\\ 
T_5: & E_5=\left\{ e_1^5=\left( 1,2,3,4,5\right) ,e_2^5=\left(2,3,4,5,6\right) \right\}\\
& S_{1,2}^5=e_1^5\cap e_2^5=\left\{ 2,3,4,5\right\},\\
& \hspace{10mm}a_{1,2}^5=e_1^5-S_{1,2}^5=\left\{ 1\right\},
  b_{1,2}^5=e_2^5-S_{1,2}^5=\left\{ 6\right\},
  c_{a_{1,2}^5,b_{1,2}^5|S_{1,2}^5}=c_{1,6|2,3,4,5}.                    
\end{array}
\]

\begin{figure}[!ht]
\centering\includegraphics[bb=65 160 515 750,width=8cm]{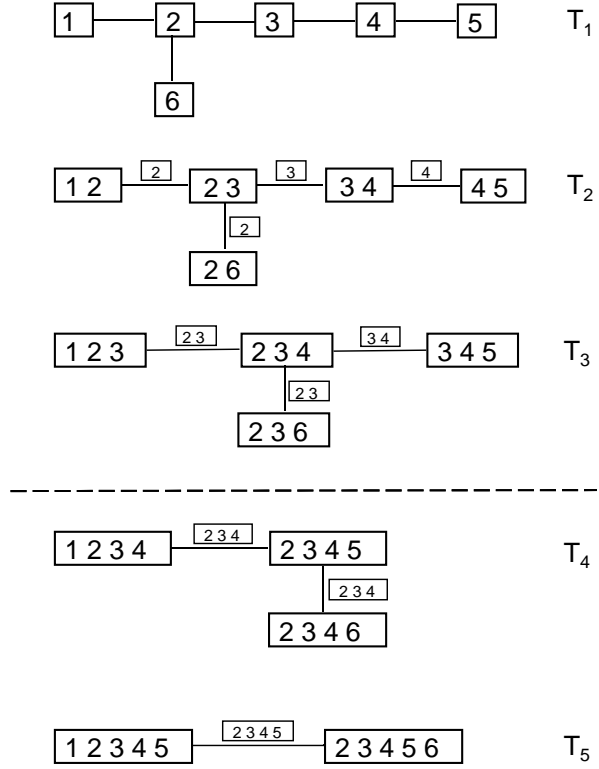}
\caption{Example for an R-vine structure on 6 variables using Definition \ref{def:def4.1}}
\label{fig:fig2}
\end{figure}

The joint probability density function of $\mathbf{X=}\left( X_1,\ldots
,X_6\right)$ can be expressed by Theorem \ref{theo:theo4.2} as follows:
\[
\begin{array}{l}
f\left( x_1,x_2,x_3,x_4,x_5,x_6\right) = \\
= \left( \prod\limits_{i=1}^6f\left(x_i\right) \right) c_{1,2}\left( F_1\left( x_1\right) ,F_2\left( x_2\right)\right) \cdot c_{2,3}\left( F_2\left( x_2\right) ,F_3\left( x_3\right)\right) \cdot c_{2,6}\left( F_2\left( x_2\right) ,F_6\left( x_6\right)\right) \\
\cdot c_{3,4}\left( F_3\left( x_3\right) ,F_4\left( x_4\right) \right)\\
\cdot c_{4,5}\left( F_4\left( x_4\right) ,F_5\left( x_5\right) \right)\\
\cdot c_{1,3|2}\left( F_{1|2}\left( x_1|x_2\right) ,F_{3|2}\left( x_3|x_2\right)\right) \\
\cdot c_{3,6|2}\left( F_{3|2}\left( x_3|x_2\right) ,F_{6|2}\left(x_6|x_2\right) \right) \\
\cdot c_{2,4|3}\left( F_{2|3}\left( x_2|x_3\right) ,F_{4|3}\left( x_4|x_3\right)\right) \\ 
\cdot c_{3,5|4}\left( F_{3|4}\left( x_3|x_4\right) ,F_{5|4}\left(x_5|x_4\right) \right) \\
\cdot c_{1,4|2,3}\left( F_{1|2,3}\left( x_1|x_2,x_3\right) ,F_{4|2,3}\left(x_4|x_2,x_3\right) \right) \\ 
\cdot c_{4,6|2,3}\left( F_{4|2,3}\left(x_4|x_2,x_3\right) ,F_{6|2,3}\left( x_6|x_2,x_3\right) \right) \\
\cdot c_{2,5|3,4}\left( F_{2|3,4}\left( x_2|x_3,x_4\right) ,F_{5|3,4}\left(x_5|x_3,x_4\right) \right) \\
\cdot c_{1,5|2,3,4}\left( F_{1|2,3,4}\left( x_1|x_2,x_3,x_4\right),F_{5|2,3,4}\left( x_5|x_2,x_3,x_4\right) \right) \\
\cdot c_{5,6|2,3,4\text{ }}\left( F_{5|2,3,4}\left( x_1|x_2,x_3,x_4\right),F_{6|2,3,4}\left( x_6|x_2,x_3,x_4\right) \right) \\
\cdot c_{1,6|2,3,4,5}\left( F_{1|2,3,4,5}\left( x_1|x_2,x_3,x_4,x_5\right)
,F_{6|2,3,4,5}\left( x_6|x_2,x_3,x_4,x_5\right) \right)
\end{array}
\]
\end{exa}

In this example we expressed the probability density function in a simplified form as in general each conditional pair copula depends on the conditioning variables (see (\ref{eq:eq9})).

Here we call the attention that our R-vine representation by a sequence of cherry-trees was also used in a recent paper by Hobaeck-Haff et al. 
(2016), Section 3. However in their paper it was not clearly declaired that this representation was introduced in our paper Kov{\'a}cs, Sz{\'a}ntai (2012b). 

In the following sections we give some theorems dealing with the relation between truncated R-vines and cherry-tree copulas.

\section{Truncated R-vine as a special case of cherry-tree copula}
\label{sec:sec3}
\smallskip

In this section we give theorems highlighting the relation between the truncated R-vine and cherry-tree copulas.
 
As the number of variables grows, the number of conditional pair-copulas grows rapidly. For example in (Dissman et al. (2013)) for 16 variables the number of 
pair-copulas involved, which have to be modeled and fitted is $120= 15+14+\cdots +2+1$. To keep such structure tractable for inference and model selection, the simplifying assumption that copulas of conditional distributions do not depend on the variables which they are conditioned on is popular. Although this assumption leads in many cases to misspecifications as it was pointed out in Acar et al. (2012) and in Hobaek Haff and Segers (2010). In Hobaek Haff et al. (2010) there are presented classes of distributions where simplification is applicable.
An idea to overcome the fitting of a large number of pair-copulas with large conditioning set is to exploit the conditional independences between the random variables. This idea was already discussed for Gaussian copulas in Kurovicka and Cooke (2006), based on the idea inspired by Whittaker (1990). However our approach is more general.

In the following remark Aas et al (2009) give the relation between conditional independences and conditional pair-copulas.

\begin{rem}\label{rem:rem5.1}
$X_i$ and $X_j$ are conditionally independent given the set of variables $%
\mathbf{X}_A, A \subset V\backslash \left\{ i,j\right\} $ if and only if
\[
c_{ij|A}\left( F_{i|A}\left( x_i|\mathbf{x}_A\right) ,F_{j|A}\left( x_j|%
\mathbf{x}_A\right) \left | \;\mathbf{x}_A \right. \right) =1. 
\]
\end{rem}

The following theorem is an important consequence of Theorem \ref{theo:theo2.9}.

\begin{thm}\label{theo:theo5.2}
If in an R-vine the conditional copula densities corresponding to the trees $T_k, T_{k+1}, \ldots, T_{d-1}$ are all equal to 1 then there exists a joint probability distribution which can be expressed only with the
conditional copula densities assigned to $T_1,\ldots ,T_{k-1}$:
\[
\begin{array}{l}
f\left(  x_{1},\ldots,x_{d}\right)  =\left[  \prod\limits_{i=1}^{d}%
f_{i}\left(  x_{i}\right)  \right]  \left[  \prod\limits_{\left(  l_{i}%
m_{i}\right)  \in E_{1}}^{{}}c_{l_{i}m_{i}}\left(  F_{l_{i}}\left(  x_{l_{i}%
}\right)  ,F_{m_{i}}\left(  x_{m_{i}}\right)  \right)  \right]  \cdot\\
\prod\limits_{l=2}^{k-1}\prod\limits_{e_{i}^{l},e_{j}^{l}\in N\left(
T_{l}\right)  }^{{}}c_{a_{i,j}^{l},b_{i,j}^{l}|S_{ij}^{l}}
\left(F_{a_{i,j}^{l}|S_{ij}^{l}}\left(  x_{a_{i,j}^{l}}|\mathbf{x}_{S_{ij}^{l}%
}\right)  ,F_{b_{i,j}^{l}|S_{ij}^{l}}\left(  x_{b_{i,j}^{l}}|\mathbf{x}%
_{S_{ij}^{l}}\right)  |\mathbf{x}_{S_{ij}^{l}}\right)\end{array}
\]
where $e_{i}^{l},e_{j}^{l}\in N\left(  T_{l}\right)  $ denotes that $e_{i}%
^{l},e_{j}^{l}$ are linked in the cherry tree $T_{l}$, and $S_{ij}^{l}%
,a_{i,j}^{l},b_{i,j}^{l},$ are defined by (\ref{eq:eq7a}) and $F_{a_{i,j}^{l}|S_{ij}^{l}}%
$ is defined in similar way as in Theorem \ref{theo:theo2.9}.
\end{thm}

The following definition of {\it truncated vine at level k} is given in
Brechmann et al. (2012).

\begin{defn}\label{def:def5.3}
A \textit{pair-wisely truncated R-vine at level k} (or truncated R-vine at level $k$) is a special R-vine copula with the property that all pair-copulas
with conditioning set equal to, or larger than $k$, are set to bivariate independence copulas.
\end{defn}

We call the attention that Brechmann denotes the first tree $T_0$. To be consistent with our earlier notations we will denote the first tree by $T_1$.

In their approach Brechmann et al. (2012), construct the truncated vines by choosing in the first $k$-trees the strongest Kendall-tau between the variables. In the last trees the pair-copulas were set to one. We claim that the strong dependences in the lower trees do not imply conditional independences in the last trees in general. This is easy to understand because of the great number of possibilities to build the last trees, starting from the same first trees. 

Another approach, which is much closer to ours, is given in Kurowicka (2011). Her idea was building trees with lowest dependence (conditional independences) in the top trees, starting with the last tree (node). Her method uses partial correlations which in case of Gaussian copula are theoretical well grounded.

There arise the following questions. What special properties has the probability
distribution, if we set to 1 the conditional copula densities associated to
the trees $T_k,\ldots ,T_{d-1}$ of its R-vine? Which are the conditional independences 
encoded in the obtained copula.

If the conditional copulas associated to the tree $T_3$ of Figure \ref{fig:fig3}:
\begin{equation}\label{eq:eq10}
\begin{array}{l}
c_{1,4|2,3}\left( F_{1|23}\left( x_1|x_2,x_3\right), F_{4|2,3}\left(x_4|x_2,x_3\right) \right), \\
c_{4,6|2,3}\left( F_{1|23}\left(x_1|x_2,x_3\right), F_{4|2,3}\left( x_4|x_2,x_3\right) \right), \\
c_{2,5|3,4}\left( F_{2,5|3,4}\left( x_2|x_3,x_4\right)), F_{5|3,4}\left(x_5|x_3,x_4\right) \right), \\ 
\end{array}
\end{equation}
are equal to 1, these imply the following conditional independences between the variables:
\begin{equation}\label{eq:eq11}
\begin{array}{lll}
X_1\perp X_4|X_2,X_3; & X_4\perp X_6|X_2,X_3; & X_2\perp X_5|X_3,X_4. \\
\end{array}
\end{equation}
In this case the junction tree copula
associated to $T_3$ in Figure \ref{fig:fig3} gives the expression of the multivariate copula as a cherry-tree copula.

\begin{figure}[!ht]
\centering\includegraphics[bb=150 640 510 760,width=6cm]{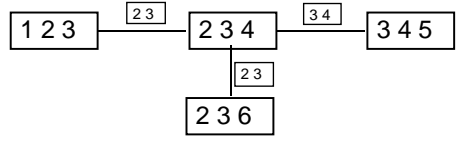}
\caption{$3$-rd order cherry junction tree}
\label{fig:fig3}
\end{figure}

The cherry-tree copula density assigned to the truncated R-vine structure in Figure \ref{fig:fig2} is:

\[
\begin{array}{l}
f\left( x_1,x_2,x_3,x_4,x_5,x_6\right) = \\
= \left( \prod\limits_{i=1}^6f\left(x_i\right) \right) \cdot c_{1,2}\left( F_1\left( x_1\right) ,F_2\left( x_2\right)\right) \cdot c_{2,3}\left( F_2\left( x_2\right) ,F_3\left( x_3\right)\right)  \\
\cdot c_{2,6}\left( F_2\left( x_2\right) ,F_6\left( x_6\right)\right) \cdot c_{3,4}\left( F_3\left( x_3\right) ,F_4\left( x_4\right) \right)
\cdot c_{4,5}\left( F_4\left( x_4\right) ,F_5\left( x_5\right) \right)\\
\cdot c_{1,3|2}\left( F_{1|2}\left( x_1|x_2\right) ,F_{3|2}\left( x_3|x_2\right)\right) 
\cdot c_{3,6|2}\left( F_{3|2}\left( x_3|x_2\right) ,F_{6|2}\left(x_6|x_2\right) \right) \\
\cdot c_{2,4|3}\left( F_{2|3}\left( x_2|x_3\right) ,F_{4|3}\left( x_4|x_3\right)\right) 
\cdot c_{3,5|4}\left( F_{3|4}\left( x_3|x_4\right) ,F_{5|4}\left(x_5|x_4\right) \right). \\
\end{array}
\]

Let us recall the following results. 
\begin{thm}\label{theo:theo5.6a} (Kov{\'a}cs, Sz{\'a}ntai (2012a))
A general $k$-width junction tree copula (see Definition \ref{def:def3.2}) can be expressed as a $k$ order cherry-tree copula. 
\end{thm}
This theorem shows why the $k$ order cherry-tree copulas are so powerful in multivariate copula modeling.

Another important result is given by the following theorem.

\begin{thm}\label{theo:theo5.5} (Kov{\'a}cs, Sz{\'a}ntai (2013))
A $k$ order cherry-tree copula can be expressed
as a $\left( k+1\right) $ order cherry-tree copula.
\end{thm}

As a consequence of Theorem \ref{theo:theo5.5} we have the following theorem.

\begin{thm}\label{theo:theo5.7} (Kov{\'a}cs, Sz{\'a}ntai (2013))
Any copula having a structure of truncated vine at level $k$ is a $k$ order
cherry-tree copula.
\end{thm}

\begin{rem}\label{rem:remu1}
Since truncated R-vine copula is a cherry-tree copula it is defined by
formula (\ref{eq:eq8}), where the set of clusters and separators are defined only by the
top tree. 
\end{rem}

From this follows, that the top tree is independent of the sequence of trees which
led to it. These trees can be constructed in multiple ways, i.e. there are many sequencies of cherry-trees leading to the top tree, therefore in our
opinion the good sequence of cherry trees is not necessarily the sequence which greedy way maximizes associations in the lower trees.
Rather we claim that it is more useful to choose that sequence which uses those pair-copulas which can be well modeled.

\begin{thm}\label{theo:theo5.8}
A $k$ order cherry-tree copula is a
truncated R-vine copula if and only if its separators define a $(k-1)$ order
cherry-tree.
\end{thm}

\begin{proof}
The first implication is that if the separators of the tree $T_{k}$ form a $(k-1)$ order cherry-tree,
then the $k$ order cherry-tree can be expressed as a truncated R-vine.
For this statement we give a constructive proof by the following algorithm.

We will show that there exists a sequence of cherry-trees which leads to the given $k$-th order cherry-tree. This means that the $k$-th order cherry-tree is a truncated R-vine at level $k$.

\begin{alg}\label{alg:alg5.9} 
{\bf Backward Algorithm.}

Algorithm for obtaining a truncated R-vine structure from a cherry-tree structure.

{\it Input}: A $k$ order cherry-tree graph structure, i.e. a set of clusters of size \textit{k}
and the set of separators of size \textit{k}-1 enhanced with the property that the separators define a $k-1$ order cherry-tree. 

{\it Output}: A R-vine truncated at level $k$.

We obtain recursively an $\left( m-1\right) $ width cherry-tree
from an $m$-width cherry-tree, for $m=k, \ldots, 1$ by the following two steps:

\begin{itemize}
\item  Step 1. The separators of the $m$-width cherry-tree will be the
clusters in the $(m-1)$-width cherry-tree, which will be linked if between them
is one cluster in the $m$-width cherry-tree, and they are different.

\item  Step 2. The leaf clusters (those clusters which contain a simplicial
node) are transformed into $(m-1)$-width clusters, by deleting one node which is
not simplicial. We emphasize here that it is essential to delete the same node from all leaf clusters which are connected to the same cluster. This guaranties that the $m-1$ order cherry-tree structure obtained is enhanced with the property that its separators define an $m-2$ order cherry-tree. The $m-1$-width cluster obtained in this way will be connected to one of
the clusters obtained in Step 1, which was the $m-1$-width separator linked to it in the 
$m$-width cherry-tree.
\end{itemize}
\end{alg}

An application of this algorithm can be seen in Figure \ref{fig:fig4}.

\begin{figure}[!ht]
\centering\includegraphics[bb=50 70 520 760,width=8cm]{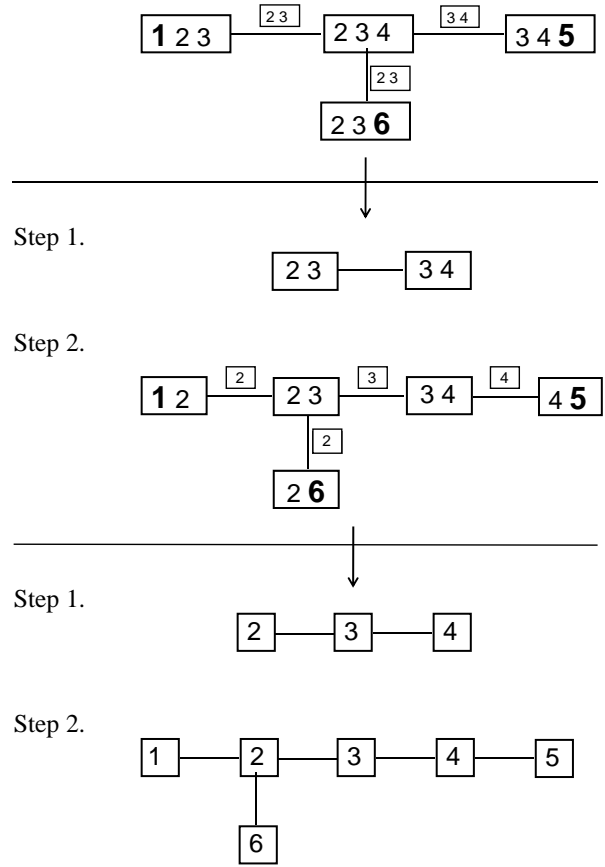}
\caption{Application of Algorithm 1 to a given 3-rd order cherry-tree in order to obtain a truncated R-vine at level 3 which leads to it.}
\label{fig:fig4}
\end{figure}

Now we prove the other implication: If the $k$ order cherry-tree copula can be
expressed by an R-vine truncated at level $k$ then the separators define
a $(k-1)$ order cherry-tree. We prove this by proving an equivalent
statement. If the separators do not define a $(k-1)$ order
cherry-tree, then it cannot be expressed as an R-vine truncated at
level $k$. We prove this on the example in Figure \ref{fig:fig5}.

\begin{figure}[!ht]
\centering\includegraphics[bb=80 660 500 780,width=8cm]{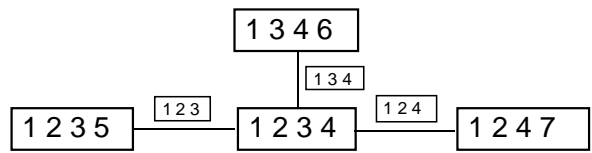}
\caption{A 4 order cherry-tree copula which cannot be achieved as a truncated R-vine}
\label{fig:fig5}
\end{figure}

Let $T_4$ be the 4 order cherry-tree in Figure \ref{fig:fig5}.  Its separators 
do not define a 3-rd order cherry-tree. We will prove, that there does not exist any 3-rd order
cherry-tree $T_3$ with the property that $T_4$ can be obtained from it by
Defintion \ref{def:def4.1}, which means that there does not exist a truncated R-vine structure
which leads to it.  

We will show that there does not exist a $T_3$ cherry-tree
with clusters in $E_3$, such that the clusters in $E_4=\left\{ \left(
1,2,3,5\right) ,\left( 1,3,4,6\right) ,\left( 1,2,3,4\right) ,\left(
1,2,4,7\right) \right\}$, could be obtained by the union of two linked
clusters belonging to $E_3$.

There are two possibilities:
\begin{itemize}
\item[1)] The clusters $\left( 1,2,3\right) ,\left( 1,3,4\right) ,\left(
1,2,4\right) $ are clusters of $E_3$. This cannot be the case because the running intersection property
could not be fulfilled.

\item[2)] At least one of these clusters is not in $E_3$. 
Without loss of generality let us suppose that $\left( 1,2,3\right)$ is not a cluster in $E_{3}$.
This means that one of the pairs $(1,2),(2,3)$ and $\left( 1,3\right)$ are not connected in $T_{3}$.

Without loss of generality let us suppose that $\left( 1,2\right) $ are not
linked. By Definition \ref{def:def4.1} this means that $\left( 1,2,3,5\right) $ in $T_4$ can be obtained from the union of
$\left(1,3,5\right) $ and $\left( 2,3,5\right)$ which are linked in $T_3$.

Now there are two sub-cases again. 
\begin{itemize}
\item[2a)]
$(1,3,5)$ is a leaf cluster (only one cluster is connected directly to it, in this case $(2,3,5)$). 
This leads to contradiction because $1$  appears in at least one of other third
order clusters, contained for example in $(1,2,4,7)$. 
\item[2b)]
$(1,3,5)$ is not a leaf cluster. In this case it is linked to another cluster
by $\left( 1,3\right) ,(1,5)$ or $(3,5)$. The other cluster has the form 
$\left( 1,3,k\right) ,(1,5,k)$ or $(3,5,k)$, with $k\in \left\{ 4,6,7\right\} $. 
By Defintion \ref{def:def4.1} the clusters of $T_{4}$ are obtained by the union of
the linked clusters in $T_{3}$. So by taking the union of $(1,3,5)$ with any of the
clusters $\left( 1,3,k\right) ,(1,5,k)$ or $(3,5,k)$ we obtain a 4 order
cluster $(1,3,5,k)$, which also leads to contradiction beause only one
cluster of $T_{4}$ contains $5$ but in this case we would have two clusters $\left( 1,2,3,5\right)$ and $%
(1,3,5,k)$ both of them containing $5$.
\end{itemize}
\end{itemize}
\end{proof}

\begin{defn}\label{def:def3.10}
The truncated R-vine obtained by the Algorithm \ref{alg:alg5.9} (Backward Algorithm) started from a given cherry-tree as the top tree is called cherry-vine structure.
\end{defn}

\begin{rem}\label{rem:rem5.6}
Algorithm \ref{alg:alg5.9} can result more cherry-vine structures as in Step 2 we may proceed in different directions.
\end{rem}

\begin{rem}\label{rem:rem5.10} 
As it can be seen from the proof of Theorem \ref{theo:theo5.8} not every cherry-tree copula is a truncated vine copula.
\end{rem}

\begin{lem}\label{lemm:lemm3.1}
A necessary and sufficient condition for a cherry-tree copula to be a
truncated R-vine copula is that each cluster has to be connected to its neighbors with
at most two different separators.
\end{lem}

\begin{proof}
First the necessity. If a cluster is connected to its neighbors by more than
two different separators then the cherry-tree copula is not a truncated
R-vine, see Figure \ref{fig:fig5}.

Now the sufficiency. We want to prove that if a cherry-tree copula is a truncated
R-vine copula then each cluster has at most two different separators. This is
equivalent to the following. If a cluster of a cherry-tree has more than to different
separators then it is not a truncated R-vine.

Let $\left\{  i_{1},\ldots,i_{k}\right\}$ be an arbitrary cluster of a $k$
order cherry-tree. Let us suppose that it is connected to its neighbors by
three different separators. Without loss of generality let us denote these
separators as follows: $S_{\backslash i_{1}}=\left\{  i_{2},\ldots,i_{k}\right\}  ,\quad
S_{\backslash i_{2}}=\left\{  i_{1},i_{3}\ldots,i_{k}\right\}  ,\quad S_{\backslash i_{3}}%
=\left\{  i_{1},i_{2},i_{4}\ldots,i_{k}\right\}  $. Any two of them will
define a ($k-1$) order cherry-tree but all three will not, since for any
permutation of the three separators, there will be an element which do not
fulfill the running intersection property. For example if the following
connection is proposed%
\[
S_{\backslash i_{1}} - S_{\backslash i_{2}} - S_{\backslash i_{3}}
\]
then $i_{2}$ occurs in the first and last set, but not in the set on the path
between them.
\end{proof}

This relationship was also discussed in a recent paper by Hobaek-Haff et al. (2016).

Lemma \ref{lemm:lemm3.1} can be used for checking wether a cherry-tree copula is or is not a truncated R-vine copula.

At this point we can conclude that the truncated vine at level $k$ is a $k$ order cherry-tree copula, but not every $k$ order cherry-tree copula can be obtained as a
truncated vine at level $k$. At a given level $k$ the number of cherry-tree copulas is much larger than the number of truncated R-vine copulas.

\section{Methods for constructing cherry-tree copulas and truncated R-vine copulas}
\label{sec:sec4}
\smallskip

Regarding to Kurowicka's approach where she says: ''We start building the
vine from the top node, and progress to the lower trees, ensuring that
regularity condition is satisfied and partial correlations corresponding to
these nodes are the smallest. If we assume that we can assign the
independent copula to nodes of the vine with small absolute values of
partial correlations, then this algorithm will be useful in finding an
optimal truncation of a vine structure'', we claim, that there are copulas
which have conditional independences in the top trees ($m \ge k$), however they have not
a truncated R-vine structure at level $k$.

We may have the following decomposition from the last node backward, which leads to 
the cherry-tree which is not truncated R-vine in Figure \ref{fig:fig6}.

\begin{figure}[!ht]
\centering\includegraphics[bb=70 420 570 770,width=9 cm]{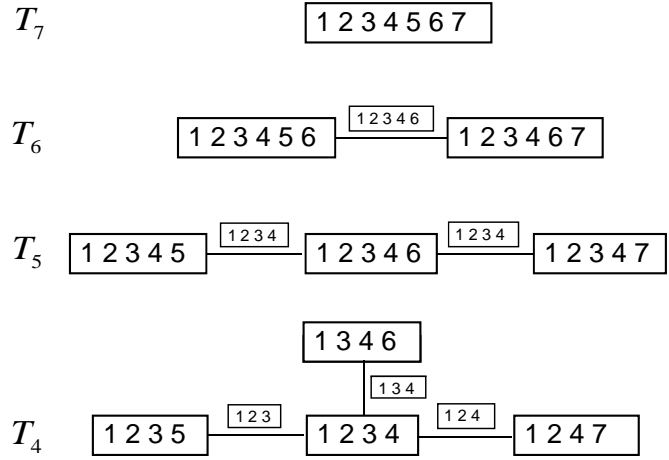}
\caption{The backward decomposition which leads to a 4-th order cherry-tree, but not an R-vine truncated at level 4.}
\label{fig:fig6}
\end{figure}

At this end the following question may arise. How can we express in general a 
cherry-tree copula by pair-copulas and conditional pair-copulas. We have two possibilities:

First, if the cherry-tree copula is a truncated R-vine copula (the separators form a tree
as we have seen in Theorem \ref{theo:theo5.8}), 
then use Algorithm \ref{alg:alg5.9} to achieve a
truncated vine structure, to which will be assigned the pair-copulas. In this case its formula is the following:

\[
\begin{array}{l}
\dfrac{\prod\limits_{K\in {\mathcal C}_{ch}
}c_{K}\left(  F_{\mathbf{X}_{K}}(\mathbf{x}_{K})\right)  }{\prod\limits_{S\in
\mathcal{S}_{ch}}\left[  c_{S}\left(  F_{\mathbf{X}_{S}}(\mathbf{x}_{S})\right)  \right]
^{\nu_{S}-1}}
= \left[  \prod\limits_{\left(  l_{i}%
m_{i}\right)  \in E_{1}}^{{}}c_{l_{i}m_{i}}\left(  F_{l_{i}}\left(  x_{l_{i}%
}\right)  ,F_{m_{i}}\left(  x_{m_{i}}\right)  \right)  \right]  \cdot\\
\prod\limits_{l=2}^{k-1}\prod\limits_{e_{i}^{l},e_{j}^{l}\in N\left(
T_{l}\right)  }^{{}}c_{a_{i,j}^{l},b_{i,j}^{l}|S_{ij}^{l}}
\left(F_{a_{i,j}^{l}|S_{ij}^{l}}\left(  x_{a_{i,j}^{l}}|\mathbf{x}_{S_{ij}^{l}%
}\right)  ,F_{b_{i,j}^{l}|S_{ij}^{l}}\left(  x_{b_{i,j}^{l}}|\mathbf{x}%
_{S_{ij}^{l}}\right)  |\mathbf{x}_{S_{ij}^{l}}\right).
\end{array}
\]
where $e_{i}^{l},e_{j}^{l}\in N\left(  T_{l}\right)  $ denotes that $e_{i}%
^{l},e_{j}^{l}$ are linked in the cherry tree $T_{l}$, and $S_{ij}^{l}%
,a_{i,j}^{l},b_{i,j}^{l},$ are defined by (\ref{eq:eq7a}) and $F_{a_{i,j}^{l}|S_{ij}^{l}}%
$ is defined in similar way as in Theorem \ref{theo:theo2.9}.

Second, if the cherry-tree copula is not a truncated R-vine copula the following theorem will be powerful for solving this problem. 

Before stating the following theorem it is important to call the attention on the following. Because the same cherry-tree can be represented graphically in multiple ways (when $\nu_S$ is greater than $1$), it is important to start with a cherry-tree where the same separators link the same cluster to other clusters. This means that from all clusters linked by the same separator we choose one and all the others will be linked to it. In this way all the other clusters will be neighbors of the chosen one. 

\begin{thm}\label{theo:theo5.9}
Starting from any $k$ order cherry-tree copula the $(k+1)$ cherry-tree copula obtained by joining the neighbor clusters via Definition \ref{def:def4.1} will be a truncated
R-vine copula.
\end{thm}

\begin{proof}
Since we are interested only in the number of the different separators, we
may suppose without loss of generality that all separators have multiplicity one. 

\begin{figure}
  \includegraphics[bb=10 205 530 530,width=14cm]{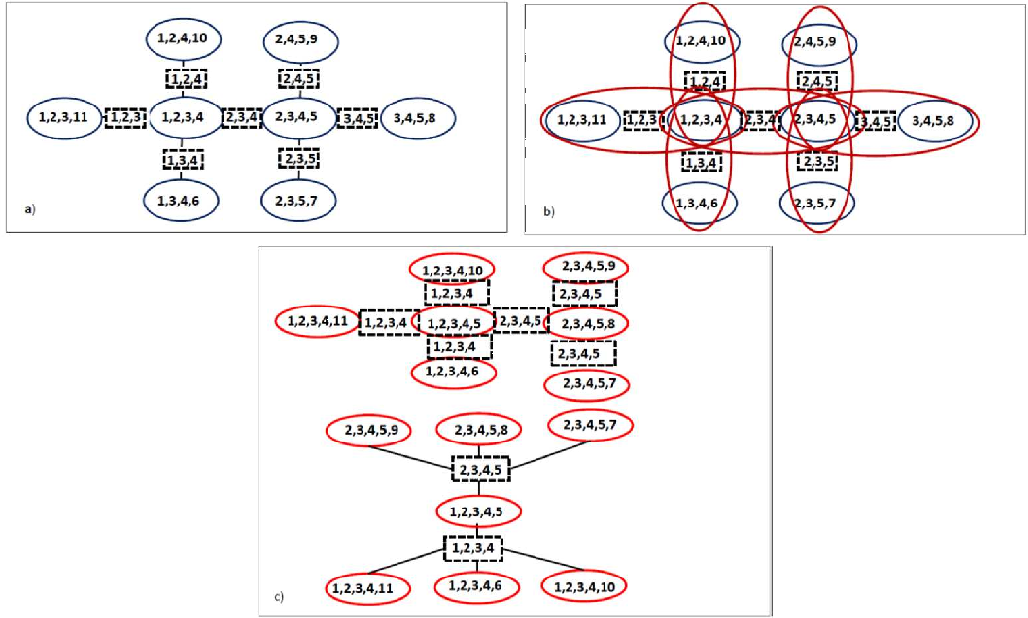}
\caption{a) A cherry-tree copula which is not truncated R-vine , b) Joining the neighboring clusters via Definition \ref{def:def4.1}, c) The obtained 5-th order cherry-tree which is a truncated R-vine in two representations}
\label{fig:fig7}       
\end{figure}

We have two cases. 

In the first case the $k$ order cherry-tree copula is already a truncated R-vine copula,
then by Definition \ref{def:def4.1} the obtained $k+1$ order cherry-tree copula is also a
truncated R-vine.

In the second case we suppose that the $k$ order cherry-tree copula is not a truncated R-vine copula. This means by
Theorem \ref{theo:theo5.9} that the set of separators do not define a 
cherry-tree. Lemma \ref{lemm:lemm3.1} implies that there exists at least one cluster 
$C_{k}^{\ast}=\left\{ i_1, i_2, \ldots ,i_k\right\}$ which is connected to its neighbors by more than two different
separators. 

By joining the neighbor clusters in the $k$
order cherry-tree using Definition \ref{def:def4.1} we obtain a $(k+1)$ order 
cherry-tree in which the
separators, correspond to the clusters in the $k$ order tree. Let us denote
by $C_{k+1}^{\ast }$ a cluster of the $\left( k+1\right) $ order 
cherry-tree obtained by joining two clusters, such that one of them is 
$C_{k}^{\ast }$. Now $C_{k+1}^{\ast }$ is separated by $C_{k}^{\ast }$ 
and at most by another different separator. For a better understanding let $C_{k}^{\ast }=\left\{  1,2,3,4 \right\}$  see picture a) of Figure \ref{fig:fig7}. Then in picture b) the process of joining the neighbor clusters is presented. In picture c) of Figure \ref{fig:fig7} we have $C_{k+1}^{\ast }=\left\{  1,2,3,4,5 \right\}$ which has two different separators connected to it, one of them is 
$C_{k}^{\ast }=\left\{  1,2,3,4 \right\}$  the other is $\left\{  2,3,4,5 \right\}$.

We emphasize here, that by joining any two $k$ order clusters in the $k$
order cherry-tree we obtain a $(k+1)$ order cluster which will have at most two neighbor
clusters in the $(k+1)$ order tree, connected to it by different separators.
By Lemma \ref{lemm:lemm3.1} the $(k+1)$ tree obtained from the $k$
order tree by Definition \ref{def:def4.1}, will have a truncated R-vine structure.
\end{proof}

We conclude this section with the following. It is easy to see how restrictive is
to search for truncated R-vines only by building from bottom up the first
trees in a greedy way. A truncated R-vine is defined by its top tree, the
lower trees can be chosen arbitrarily only by fulfilling the condition given in Definition \ref{def:def4.1}.

The idea is searching good fitting cherry-tree copulas and then to express it by a truncated R-vine copula. We proved in Theorem \ref{theo:theo5.9} that any
cherry-tree copula can be transformed into a truncated R-vine which can be
reached by using the Backward Algorithm.

\section{Conclusions}
\label{sec:sec6}

In modeling multivariate probability
distribution, an important task is to exploit some conditional independences
existing between the random variables. We introduced 
in Sz\'{a}ntai, Kov\'{a}cs (2012) the discrete
cherry-tree probability distributions, then 
in Kov\'{a}cs, Sz\'{a}ntai (2012a) the cherry-tree copulas.
The results of the present paper link the cherry-tree copula to the truncated
R-vine which makes possible the use of cherry-tree structures in modeling
continuous probability distributions, too.

If the number of variables grows the general R-vine copula modeling gets
untractable. A method to overcome this problem is exploiting the conditional
independences between the variables. The cherry-tree copulas are able to
exploit these conditional independences. Another model containing
conditional independences is the truncated R-vine. In the literature it was
mainly fitted in greedy way from bottom to up. 

In this paper we clarify the relation between the
cherry-tree copula and the truncated R-vine copula. The cherry-tree copula is more
general than the truncated R-vine copula, but the truncated R-vine copula has the powerful
property that it can be expressed by pair-wise copulas and pair-wise conditional
copulas. We proved that a $k$-th order cherry-vine copula can be either expressed as a truncated R-vine copula at level $k$ (by using the Backward Algorithm) or transformed into a $k+1$ order cherry-tree copula which can be expressed by a truncated vine copula at level $k+1$ (Theorem \ref{theo:theo5.9}). In this way the
cherry-tree copula gets also this powerful property.

In Kov{\'a}cs, Sz{\'a}ntai (2012a) we proved that any general $k$-width junction tree copula can be
embedded in a $k$ order cherry-tree copula. This shows the power of
cherry-tree copulas related to general junction tree copulas.

Because the truncated R-vine copula is completely characterized by the top tree (at a given level), in our opinion finding good truncated R-vines should be started by finding a good top tree (cherry-tree). A possible method for this, starting from a data set, is presented in Kov\'{a}cs, Sz\'{a}ntai (2010). Then one can construct the sequence of the cherry-trees which leads to it. This is the so called cherry-vine structure. 

We believe our approach may open a new perspective in modeling continuous multivariate probability distributions by exploiting the conditional independences between the components of the random vector. We challenge the vine copula community to search for good models from this perspective.

\smallskip
{\footnotesize
\hspace*{0.5cm}

\begin{minipage}[t]{8cm}$$\begin{array}{l}
\mbox{Edith Kov\'{a}cs -- Department of Differential Equations,}\\
\mbox{Budapest University of Technology and Economics},\\
 \mbox{M\H{u}egyetem rkp. 3., Budapest, 1111 HUNGARY}\\
\mbox{E-mail: kovacsea@math.bme.hu}\\ \\
\mbox{Tam\'{a}s Sz\'{a}ntai -- Department of Differential Equations,}\\
\mbox{Budapest University of Technology and Economics},\\
 \mbox{M\H{u}egyetem rkp. 3., Budapest, 1111 HUNGARY}\\
\mbox{E-mail: szantai@math.bme.hu}\end{array}$$
\end{minipage}}

\end{document}